\documentclass[envcountsame]{llncs}

\usepackage{amsmath,amssymb}

\usepackage[colorlinks, linkcolor=black, citecolor=black, filecolor=black]{hyperref}

\spnewtheorem*{remarkx}{Remark} {\itshape}{\rmfamily}
\spnewtheorem{observation}[theorem]{Observation} {\bfseries}{\upshape}
\spnewtheorem{assumption}{Assumption} {\bfseries}{\itshape}
\spnewtheorem*{xmpl}{Example} {\itshape}{\rmfamily}

\def\clap#1{\hbox to 0pt{\hss#1\hss}}
 \def\mathrlap{\mathpalette\mathrlapinternal} 

\def\mathrlapinternal#1#2{\rlap{$\mathsurround=0pt#1{#2}$}}

\newcommand{\paths}{\mathcal{P}}
\newcommand{\pathsT}{\paths_T}
\newcommand{\pathsTw}{\paths_T^\ast}
\newcommand{\T}{\mathcal{T}}
\newcommand{\NP}{$N\!P$}
\newcommand{\arr}{\gamma}

\title{Abstract flows over time: A first step towards solving dynamic packing problems\thanks{This work was supported by Deutsche Forschungsgemeinschaft (DFG) as part of the Priority Program ``Algorithm Engineering'' (1307), by DFG Research Center \textsc{Matheon} ``Mathematics for key technologies'' in Berlin, and the Berlin Mathematical School. An extended abstract of this article will be published in \emph{Algorithms and Computation: 23rd International Symposium, ISAAC 2012}.}}
\author{Jan-Philipp W. Kappmeier \and Jannik Matuschke \and Britta Peis}
\institute{TU Berlin, Institut f\"{u}r Mathematik, Stra{\ss}e des 17. Juni 136, 10623 Berlin, Germany\\ \email{\{kappmeier,matuschke,peis\}@math.tu-berlin.de}\\
\vspace{0.5cm}
Preprint 001-2012 (revised version)\\ November 12, 2012\vspace{-0.5cm}}
 
\begin{document}

\maketitle

\begin{abstract}
\emph{Flows over time}~\cite{ford62} generalize classical network flows by introducing a notion of time. Each arc is equipped with a transit time that specifies how long flow takes to traverse it, while flow rates may vary over time within the given edge capacities. In this paper, we extend this concept of a dynamic optimization problem to the more general setting of \emph{abstract flows}~\cite{hoffman74}. In this model, the underlying network is replaced by an abstract system of linearly ordered sets, called ``paths'' satisfying a simple switching property: Whenever two paths $P$ and $Q$ intersect, there must be another path that is contained in the beginning of $P$ and the end of $Q$.
 
We show that a maximum abstract flow over time can be obtained by solving a weighted abstract flow problem and constructing a temporally repeated flow from its solution. In the course of the proof, we also show that the relatively modest switching property of abstract networks already captures many essential properties of classical networks.
\end{abstract}

\section{Introduction}

Time plays a crucial role in many applications of combinatorial optimization, e.g., in the context of transportation, communication, or productional planning. Therefore, extending classical problem formulations by a temporal dimension is of particular interest. So far the most prominent example in this direction is the concept of \emph{flows over time} -- also called ``dynamic flows'' in the literature -- which was first introduced and investigated by Ford and Fulkerson~\cite{ford62}. A key challenge in the context of flows over time is that an explicit specification of all flow values at each time step leads to an output that is exponential in the input size.
Ford and Fulkerson resolved this issue by showing that the maximum flow over time problem allows for a so-called \emph{temporally repeated} solution, which can be obtained by solving a single static flow problem. 
Since then, numerous results on different variants of flow over time problems have emerged. Outstanding results include \cite{fleischer1998,hoppe2000,klinz2004}, see~\cite{skutella09} for a general survey. 

Network flows can be interpreted as a special case of packing problems: we try to pack the capacitated edges of the graph by assigning flow values to the source-sink-paths. Given the impact of Ford and Fulkerson's result, which spawned a whole theory of flows over time, one now might ask how the concept of time can be extended to other packing problems. A first natural candidate are generalizations of static network flows, as, e.g., \emph{abstract flows}. The notion of abstract flows goes back to Hoffman~\cite{hoffman74}, who observed that Ford and Fulkerson's original proof of the max flow/min cut theorem~\cite{ford54} does not use the underlying network structure directly but only exploits one particular property of the path system, the so-called \emph{switching property}. Hoffman succeeded in showing that packing problems defined on general set systems (called \emph{abstract networks}) with this switching property are totally dual integral (TDI). These structural results were later complemented by the combinatorial primal-dual algorithms of Martens and McCormick~\cite{mccormick96,martens08}. Inspired by Hoffman's work, further abstractions based on uncrossing axioms have been proposed and corresponding TDI results established, e.g., lattice polyhedra~\cite{hoffman1978} or switchdec polyhedra~\cite{gaillard97},  see~\cite{schrijver1984} for a survey.
In light of these generalizations, abstract flows appear to serve as an ideal first stepstone in our endeavour towards dynamic formulations of more general packing integer programs.
\vspace{-0.5cm}
\subsubsection{Our contribution}
In this paper, we introduce and investigate \emph{abstract flows over time} and show how a temporally repeated abstract flow and a corresponding minimum cut can be computed by solving a single static weighted abstract flow problem. This immediately leads to the max flow/min cut theorem for abstract flows over time as our main result. Although our construction resembles that of Ford and Fulkerson's original result \cite{ford62} on (non-abstract) flows over time, the proof turns out to be considerably more involved and we will need to take a detour via a relaxed version of abstract flows over time that also considers storage of flow at intermediate elements. However, our results also imply that this relaxation is not proper and there always is an optimal solution that does not wait at intermediate nodes. In the course of our proof, we also establish some interesting structural properties of abstract networks, showing that the relatively modest switching property of abstract path systems already captures many essential properties of classical networks.

\vspace{-0.5cm}
\subsubsection{Structure of this paper} In the remainder of this section, we introduce Hoffman's model of abstract flows in detail. In~\autoref{sec:time-expansion}, we show how to conduct a time expansion on this model and point out differences to the time expanded network for classical network flows by Ford and Fulkerson~\cite{ford62}. In~\autoref{sec:temp-rep}, we will show how to construct the temporally repeated abstract flow and a corresponding minimum abstract cut of same value. In order to validate feasibility of this cut, we will prove the necessary properties on the structure of abstract networks in~\autoref{sec:structure}. Using these results, we can finally show in~\autoref{sec:final-proof} that the cut actually intersects all temporal paths, completing the proof of our main theorem.

\subsection*{Introduction to abstract flows}

An \emph{abstract path system} consists of a ground set $E$ of \emph{elements} and a family of \emph{paths} $\paths \subseteq 2^E$. For every $P \in \mathcal{P}$ there is an order $<_P$ of the elements in $P$. A path system is an \emph{abstract network}, if the \emph{switching property} is fulfilled: For every $P, Q \in \paths$ and every $e \in P \cap Q$, there is a path 
$$P \times_e Q \subseteq \{p \in P :~ p \leq_P e\} \cup \{q \in Q :~ q \geq_Q e\}.$$

Given an abstract network with capa\-cities $c \in \mathbb{R}_+^E$ for all elements, the \emph{maximum abstract flow problem} asks for an assignment of flow values $x \in \mathbb{R}_+^\paths$ to the paths such as to maximize the total flow value while not violating the capacity of any element. The problem can be generalized further by introducing a weight function $r \in \mathbb{R}_+^\paths$ that specifies the ``reward'' per unit of flow sent along each path. It is easy to see that allowing general weight functions renders the problem \NP-hard. Thus, the choice of weight functions is restricted to supermodular functions, i.e., we require
$$r(P \times_e Q) + r(Q \times_e P) \geq r(P) + r(Q)$$
for every $P, Q \in \paths$ and $e \in P \cap Q$. 

The dual of the maximum weighted abstract flow problem is the \emph{minimum weighted abstract cut problem}, which assigns a value $y(e)$ to every element $e \in E$ so as to cover every path according to its weight.
The two problems can be stated as follows.

\begin{minipage}{3.5cm}
   \begin{alignat*}{3}
    \max\ \ && \sum_{P \in \paths} r(P) x(P) & & \\
    \text{s.t.}\ \ && \sum_{P \in \paths : e \in P} x(P) \leq &\ c(e) & \ \ \forall e \in E \\
    && x(P) \geq &\ 0 & \ \ \forall P \in \paths
    \end{alignat*}
\end{minipage}
\hspace{0.8cm}
\begin{minipage}{3.5cm}
\begin{alignat*}{3}
 \min\ \ && \sum_{e \in E} c(e) y(e) & & \\
 \text{s.t.}\ \ && \sum_{e \in P} y(e) \geq &\ r(P) & \ \ \forall P \in \paths \\
 && y(e) \geq&\ 0 & \ \ \forall e \in E
\end{alignat*}
\end{minipage}
\vspace{0.3cm}

Hoffman~\cite{hoffman74} showed that for every integral supermodular weight function, the abstract cut LP is totally dual integral. This implies a generalized version of Ford and Fulkerson's max flow/min cut result in two ways: On the one hand, the switching property represents a significant abstraction, allowing for more general structures. On the other hand, supermodular weight functions lead to weighted cuts, i.e., elements can appear multiple times in the cut. We will later see a useful example of such weights in the context of temporally repeated flows, which also yields an intuitive interpretation of these cut values.

Hoffman's structural result was extended by McCormick~\cite{mccormick96}, who presented a combinatorial algorithm that solves the unweighted version ($r \equiv 1$) of the maximum abstract flow problem in time polynomial in $|E|$, if the abstract network is given by a separation oracle for the abstract cut LP (in the unweighted case, this is equivalent to deciding whether a given set of elements contains a path or not). Later, Martens and McCormick~\cite{martens08} extended this result and presented an algorithm that also solves the weighted case. 

While these results indicate that the switching property is the essential force behind max flow/min cut and similar total dual integrality results for flow based problems, we want to close this section by pointing out an example that shows how abstract networks actually may differ from classical networks. In classical networks, if two paths $P$ and $Q$ both intersect a third path $R$, then there either is a path from the beginning of $P$ to the end of $Q$ or the other way around. The following example shows that this is not true in abstract networks, even in cases where the switching property preserves the order of intersecting abstract paths.

\begin{xmpl}\label{ex:abstract-network}
Consider the abstract network $(E, \paths)$ with $E = \{1, 2, 3, 4, a, b, c, d\}$ and $\paths = \{(1, 2, 3, 4), (a, 2, c), (b, 3, d), (1, c), (1, d), (a, 4), (b, 4)\}$. Although both $(a, 2, c)$ and $(b, 3, d)$ intersect the path $(1, 2, 3, 4)$, there is neither a path that starts with $a$ and ends with $d$ nor one that starts with $b$ and ends with $c$.
\end{xmpl}

\section{Time expansion of abstract networks}\label{sec:time-expansion}

Time plays an important role in many application areas of network flows. Flow rates can vary over time, and flow also takes time to travel within the network. One concept to capture these temporal effects is the so-called time expanded network introduced by Ford and Fulkerson~\cite{ford62}. The basic idea is to introduce multiple copies of the nodes in the network, one for each point in time. Then arcs connect copies of vertices according to their travel time.
We extend this concept to the world of abstract flows by introducing the time expansion of an abstract network. In the spirit of Ford and Fulkerson's idea, we will introduce multiple copies of the abstract network. In contrast to the classical case however, not copies of individual arcs but of whole paths will be introduced.

The \emph{time expansion} of an abstract network consists of a (static) abstract network with capacities $c \in \mathbb{R}_+^E$, \emph{transit times} $\tau \in \mathbb{Z}_+^E$ and a \emph{time horizon} $T \in \mathbb{Z}_+$. 	The time from $0$ to $T$ is discretized into $T$ intervals $[0, 1), \ldots, [T-1, T)$ which we identify with the set of their starting times $\T := \{0, \ldots, T-1\}$. For each interval, a copy of the ground set $E$ is introduced, i.e., the time expanded ground set is
$E_T := E \times \T$.

A \emph{temporal path} is denoted by $P_t$, where $P$ is a path of the underlying static abstract network and $t \in \T$ specifies the starting time of the path. Flow sent along the temporal path $P_t$ enters element $e$ at time $t + \sum_{p \in (P, e)} \tau(e)$,
which is the time it needs for traversing all preceeding elements plus the initial offset of the path. Accordingly, we identify $P_t$ with the set of its temporal elemtents by defining 
$$\textstyle P_t := \left\{(e, \theta) \in E_T :~ e \in P,~ \theta = t + \sum_{p \in (P, e)} \tau(p) \right\}.$$
The \emph{arrival time} of the temporal path $P_t$ is $t+ \sum_{e \in P} \tau(e)$, i.e., the time at which the flow arrives the end of the path. Since all flow is supposed to arrive its destination within the time horizon, we only allow copies of paths with a maximum arrival time of $T-1$, which is the final element of $\mathcal{T}$. Thus, the set of temporal paths is defined by
$$\textstyle \paths_T := \left\{P_t :~ P \in \paths,\;t \in \T,~ t + \sum_{p \in P} \tau(p) < T \right\}.$$

We now can define the maximum abstract flow over time problem in analogy to the (static) maximum abstract flow problem. An \emph{abstract flow over time} is an assignment $x: \pathsT \rightarrow \mathbb{R}_+$ of non-negative flow values to all temporal paths. It is feasible if and only if the capacity of every element at every point in time is respected. 
The maximum abstract flow over time problem asks for an abstract flow over time that maximizes the total value of the flow:
   \begin{alignat*}{3}
    \max\quad && \sum_{P_t \in \pathsT} x(P_t) & & \\
    \text{s.t.}\quad && \sum_{\mathrlap{P_t \in \pathsT :~ (e, \theta) \in P_t}}\ x(P_t) & \leq c(e) & \quad \forall e \in E,~ \theta \in \T \\
    && x(P_t) & \geq 0 & \quad \forall P_t \in \pathsT.
    \end{alignat*}

In analogy to the static case, the maximum value of an abstract flow over time can be bounded by an \emph{abstract cut over time}, i.e., a subset $C \subseteq E_T$ of the time expanded ground set such that for each $P_t \in \pathsT$ the set $P_t \cap C$ is nonempty.
	
\newcounter{count-lem:cut-is-ub}
\addtocounter{count-lem:cut-is-ub}{\arabic{theorem}}
\begin{lemma}\label{lem:cut-is-ub}
Let $x$ be an abstract flow over time and let $C$ be an abstract cut over time. Then
$\sum_{P_t \in \pathsT}x(P_t) \leq \sum_{(e, \theta) \in C} c(e)$.
\end{lemma}
\begin{proof}
As the cut contains an element of every temporal path and the capacity constraints are respected at every point in time, we get
$$\sum_{P_t \in \pathsT} x(P_t) \leq \sum_{(e, \theta) \in C} \sum_{\mathrlap{\ \ \ \ P_t: (e, \theta) \in P_t}} x(P_t) \leq \sum_{e \in C} c(e).\hspace{1cm}$$
\vspace{-1.3cm}\\ 
\qed
\vspace{0.5cm}
\end{proof}

\begin{remarkx}(Time expansion of an abstract network vs. time expanded network)
While the time expansion of abstract networks as defined above is similar to the notion of a \emph{time expanded network} as defined by Ford and Fulkerson~\cite{ford62} for classical network flows, the two definitions are not quite identical.
Time expanded networks are based on the arc formulation of network flows. They are constructed by introducing copies of both the nodes and arcs of the underlying static network and adjusting the end points of the arcs according to their transit times. By construction, the resulting structure is guaranteed to be a network again. Unfortunately, there is no correspondence to the arc formulation for abstract flows -- their definition is inherently tied to the path system, which does not allow for local concepts such as flow conservation at a particular element. Our model of time expansion therefore introduces copies of each path as a whole. In contrast to time expanded networks, the time expansion of an abstract network is \emph{not} an abstract network in general, as can be seen in the following example.
\end{remarkx}

\begin{xmpl}\label{ex:switch_violated}
Let $E = \{s, a, b, t\}$ and $\paths = \{P, Q, R, S\}$ with $P = (s, a, b, t)$, $Q = (s, b, a, t)$, $R = (s, a, t)$, and $S = (a, b, t)$. It is easy to verify that $\paths$ in fact fulfills the switching property. Now assume all elements have unit transit times, i.e., $\tau \equiv 1$. The temporal paths $P_0$ and $Q_1$ intersect in the element $(b, 2)$. However, there is no temporal path in $\pathsT$ that can be constructed from the elements $\{(s, 1), (b, 2), (t, 4)\}$, as there is a ``time gap'' between $(b, 2)$ and $(t, 4)$. Thus, the time expansion violates the switching property.
\end{xmpl}
In view of this example, it is not even clear  whether max flow/min cut results are still valid in the context of abstract flows over time or how far existing algorithms for abstract flow problems can be applied to the time expansion of the abstract network. Fortunately, the proof of our main result in the following sections will dissipate these concerns. 

\begin{theorem}[Abstract max flow/min cut over time]\label{thm:max-flow/min-cut}
The value of a maximum abstract flow over time equals the capacity of a minimum abstract cut over time. Both a maximum flow and a minimum cut over time can be computed by solving a single (static) maximum weighted abstract flow problem.
\end{theorem}

Our proof of \autoref{thm:max-flow/min-cut} involves constructing an abstract cut over time. In order to show feasibility of this cut, we will have to introduce the possibility of waiting at intermediate elements as an important device in our proof (see \autoref{sec:structure}). 
Storage of flow at intermediate nodes plays an interesting role in the field of flows over time: While in some settings, such as the maximum flow over time problem or the NP-hard minimum cost flow over time problem, there always exist optimal solutions that do not wait at intermediate nodes~\cite{ford62,fleischer2003}, this is not true in other settings: e.g., for multi-commodity flows over time, the decision of allowing flow storage at intermediate nodes has an influence on the value of the solution and also on the complexity~\cite{hall2007,gross2012}.
In the context of abstract flows over time, our results imply that the possibility of waiting has no influence on the problem, as we prove in \autoref{sec:final-proof} that the temporally repeated solution constructed in \autoref{sec:temp-rep} is optimal even if waiting is allowed.

\begin{theorem}\label{thm:max-flow/min-cut_waiting}
	If waiting at intermediate elements is allowed, there still is a maximum abstract flow over time that does not wait at intermediate elements.
\end{theorem}

\section{Constructing a maximum abstract flow over time}\label{sec:temp-rep}

The number of paths created by applying the time expansion is linear in $T$ and thus exponential in the size of the input. Hence, even encoding a solution in the straight-forward way results in an exponentially sized output. Ford and Fulkerson~\cite{ford62} resolved this problem for the classical (non-abstract) flow over time problem by introducing so-called temporally repeated flows, i.e., a flow over time constructed by temporally repeating a static flow pattern. 

A \emph{temporally repeated abstract flow} is an abstract flow over time $x$ that is constructed from a static abstract flow $\tilde{x}$ by setting
$x(P_t) := \tilde{x}(P)$ for $P \in \paths \text{ and } 0 \leq t < T - \sum_{e \in P} \tau(e)$ and $0$ otherwise.
In other words, the static flow on each path is repeatedly sent as long as possible before the time horizon is reached. It is easy to check that feasibility of the underlying static flow implies feasibility of the temporally repeated flow.

\newcounter{count-lem:temp_rep_feas}
\addtocounter{count-lem:temp_rep_feas}{\arabic{theorem}}
\begin{lemma}\label{lem:temp_rep_feas}
A temporally repeated abstract flow derived from a feasible abstract flow is a feasible abstract flow over time.
\end{lemma}
\begin{proof}
Let $\tilde{x}$ be a feasible abstract flow and let $x$ be the corresponding temporally repeated flow. We only need to verify that $x$ obeys the capacity restrictions for every $e \in E$ and every $\theta \in \T$. In fact, observe that $(e, \theta) \in P_t$ if and only if $e \in P$ and $\theta = t + \sum_{p \in P} \tau(p)$. As the second part of this sum is constant for a fixed $P \in \paths$, there is only one specific value of $t$ for which $(e, \theta) \in P_t$. Thus
$$\sum_{\mathrlap{P_t \in \pathsT:~(e, \theta) \in P_t}} x(P_t) \leq \sum_{P \in \paths} \tilde{x}(P) \leq c(e)$$
for all $(e, \theta) \in E_T$.\qed
\end{proof}

In order to construct a maximum temporally repeated abstract flow, we first observe that flow can be sent along path $P \in \paths$ up to time $r(P) := T - \sum_{e \in P} \tau(e)$, i.e., the flow value $\tilde{x}(P)$ is repeated $r(P)$ times. Thus, the total flow value of the temporally repeated flow $x$ resulting from the static flow $\tilde{x}$ is $\sum_{P \in \paths} r(P)\tilde{x}(P)$ and a maximum temporally repeated flow corresponds to a static abstract flow that is maximum with respect to the weights $r(P)$.
It is not hard to see that the weight function defined in this way is supermodular.
\newcounter{count-obs:supermodular}
\addtocounter{count-obs:supermodular}{\arabic{theorem}}
\begin{observation}The weight function $r(P) := T - \sum_{e \in P} \tau(e)$ is supermodular.
\end{observation}
\begin{proof} By definition of $r$ we have
\[\begin{aligned}
	r(P \times_e Q) + r(Q \times_e P) =\ & T - \sum\limits_{\mathrlap{e \in P \times_e Q}} \tau(e) + T - \sum\limits_{\mathrlap{e \in Q \times_e P}} \tau(e) \\
	\geq\ & 2T - \left( \sum\limits_{\mathrlap{\ \ e \in [P,e]}} \tau(e) + \sum\limits_{\mathrlap{e \in (e,Q]}} \tau(e) \right) - \left( \sum\limits_{\mathrlap{\ \ e\in[Q,e]}} \tau(e) + \sum\limits_{\mathrlap{e\in(e,P]}} \tau(e) \right) \\ 
	=\ & 2T - \sum\limits_{e\in P} \tau(e) - \sum\limits_{e\in Q} \tau(e) \\
	=\ & r(P) + r(Q).\text{\hspace{5.8cm}}\qed
\end{aligned}\]
\end{proof}

Thus, we can solve the weighted abstract flow problem defined by these weights using the algorithm from~\cite{martens08}, yielding a (static) abstract flow $\tilde{x}^\ast$ of maximum weight and the corresponding temporally repeated flow $x^\ast$.
We will show that the value of $x^\ast$ is not only maximum among the temporally repeated abstract flows but also among \emph{all} abstract flows over time.
To this end, we now construct an abstract cut over time whose capacity matches the flow value of $x^\ast$. Let $\tilde{y}$ be an optimal solution to the dual of the static weighted abstract flow problem with the weights $r(P)$ used to construct the temporally repeated flow. Note that by~\cite{hoffman74}, we can assume $\tilde{y}$ to be integral. We will interpret the values $\tilde{y}(e)$ as the number of time steps for which element $e$ is contained in the cut. We define the time at which $e \in E$ enters the cut by setting 
$$\textstyle \alpha(e) := \min_{P \in \paths} \sum_{p \in (P, e)} (\tau(p) + \tilde{y}(p))$$
and define 
$$C := \left\{(e, \theta) \in E_T:~ \alpha(e) \leq \theta < \alpha(e) + \tilde{y}(e) \right\}.$$

\begin{theorem}\label{lem:cut-over-time}
$C$ is a feasible abstract cut over time.
\end{theorem}
The proof of \autoref{lem:cut-over-time} involves some additional results on the structure of abstract networks, which we will elaborate on in the following sections. Using LP duality, \autoref{lem:cut-over-time} immediately leads to the following corollary, which implies \autoref{thm:max-flow/min-cut}.
\newcounter{count-cor:max-flow/min-cut}
\addtocounter{count-cor:max-flow/min-cut}{\arabic{theorem}}
\begin{corollary}
The temporally repeated abstract flow $x^\ast$ is a maximum abstract flow over time, and $C$ is a minimum abstract cut over time whose capacity is equal to the flow value.
\end{corollary}
\begin{proof}
We observe that by duality,
$$\sum_{(e, \theta) \in C} c(e) = \sum_{e \in E} c(e)\tilde{y}(e) = \sum_{P \in \paths} r(P)\tilde{x}^\ast(P) = \sum_{P_t \in \pathsT} x^\ast(P_t)$$
and thus the capacity of $C$ equals the flow value of $x^\ast$.\qed
\end{proof}

\section{Waiting at intermediate elements and the structure of abstract networks}\label{sec:structure}

In order to prove that the set $C$ constructed in the preceeding section actually covers all temporal paths, we need to show that we can ensure w.l.o.g.~that the switching operation $\times_\cdot$ preserves the order of the intersecting paths. We start by showing a weaker version of this statement, asserting that we can always choose the path resulting from an application of $\times_\cdot$ in such a way that the two subpaths used for its construction are not mixed.

\begin{lemma}
\label{lem:cross}
Let $P, Q \in \paths$, $e \in P \cap Q$, then there is a path $R \subseteq [P, e] \cup [e, Q]$ such that $a \in R \cap [P, e]$ and $b \in R \setminus [P, e]$ implies $a <_R b$.
\end{lemma}
\begin{proof}
Let $P, Q \in \paths$ and $e \in P \cap Q$. Let $R$ to be a path contained in $[P, e] \cup [e, Q]$ such that $|R \setminus [P, e]|$ is minimal. By contradiction assume there is $a \in R \cap [P, e]$ and $b \in R \setminus [P, e]$ with $b <_R a$. Let $R' := P \times_a R$. Observe that $R' \subset [P, e] \cup [e, Q]$ and $R' \setminus [P, e] \subset R \setminus [P, e]$ as $a \notin R'$, contradicting the choice of $R$.\qed
\end{proof}
As a result of \autoref{lem:cross}, the following assumption is without loss of generality.
\begin{assumption}\label{ass:cross}
If $a \in P \times_e Q \cap [P, e]$ and $b \in P \times_e Q \setminus [P, e]$, then $a <_{P \times_e Q} b$.
\end{assumption}

In order to show that $\times_\cdot$ actually preserves the internal order of $P$ and $Q$, we will -- temporally -- extend our model of time expansion by allowing flow to deliberately delay its traversal at intermediate elements.

\subsubsection{Waiting at intermediate elements}
A \emph{temporal path with intermediate waiting} is denoted by $P_\sigma$, where $P \in \paths$ is a path of the underlying static abstract network and $\sigma : P \rightarrow \T$ specifies the waiting time $\sigma(e)$ before traversing element $e \in P$. Flow sent along $P_\sigma$ enters element $e$ at time
$\arr(P_\sigma, e) := \sum_{p \in (P, e)} (\sigma(p) + \tau(p)) + \sigma(e)$
which is the time it needs for traversing all preceeding elements and the time it spends waiting at those elements and at $e$ itself. Accordingly, we identify $P_\sigma$ with the set of its temporal elemtents by defining 
$$P_\sigma := \left\{(e, \theta) \in E_T :~ e \in P,~ \theta = \arr(P_\sigma, e) \right\}.$$
The set of all temporal paths with intermediate waiting is denoted by
$$\textstyle \pathsTw := \left\{P_\sigma :~ P \in \paths,\;\sigma \in \T^P,~ \sum_{e \in P} (\sigma(e) + \tau(e)) < T \right\}.$$

We will identify $P_t \in \pathsT$ with $P_{(t, 0, \dots, 0)} \in \pathsTw$. Note that the maximum abstract flow over time problem with waiting at intermediate elements is a relaxation of the maximum abstract flow over time without waiting, and the temporally repeated abstract flow $x^\ast$ defined in \autoref{sec:temp-rep} is a feasible solution to this relaxation. We will show that $C$ actually covers all paths in $\pathsTw$, and thus $x^\ast$ is optimal even if waiting is allowed. This implies that the relaxation is not proper, i.e., the possibility of waiting does not have any effect on the value of the optimal solution. 

However, the extension of the model allows us to delete certain paths from the network. Observe that if $Q$ is a strict subset of $P$, and $<_Q$ is identical to the restriction of $<_P$ to $Q$, then there always is an optimal abstract flow over time that does not use any copy of $P$ (since it can wait at intermediate elements and use $Q$ instead). Thus we can safely erase $P$ from the base network in this case (without violating the switching axiom as $Q$ can always replace $P$ as switching choice). Hence, if we allow waiting at intermediate elements, the following assumption is without loss of generality.

\begin{assumption}
\label{rem:order}
If $Q \subset P$ then there are $a, b \in Q$ with $a <_P b$ and $b <_Q a$. 
\end{assumption}

In the remainder of this section, we will show that \autoref{rem:order} implies the following lemma. As a corollary, we can assume w.l.o.g. the switching operation to preserve order.

\begin{lemma}\label{lem:no-inclusion}
There are no paths $P, Q \in \paths$ such that $Q \subset P$.
\end{lemma}

\newcounter{count-cor:order-preserving}
\addtocounter{count-cor:order-preserving}{\arabic{theorem}}
\begin{corollary}\label{cor:order-preserving}
Let $R := P \times_e Q$.
If $a, b \in R \cap [P, e]$ and $a <_P b$ then $a <_R b$.
If $a, b \in R \setminus [P, e]$ and $a <_Q b$ then $a <_R b$.
\end{corollary}
\begin{proof}\ 
\begin{itemize}
\item By contradiction assume $a, b \in [P, e] \cap R$ and $a <_P b$ but $b <_R a$. Then, by \autoref{ass:cross}, there is no $c \in R \setminus (P, e)$ with $c <_R a$. Thus $[R, a] \subseteq P$ and $R \times_b P \subset P$, contradicting \autoref{lem:no-inclusion}.
\item By contradiction assume $a, b \in R \setminus [P, e]$ and $a <_Q b$ but $b <_R a$. Then, by \autoref{ass:cross}, there is no $c \in R \cap [P, e]$ with $c >_R b$. Thus $[a, R] \subseteq Q$ and $Q \times_a R \subset Q$, contradicting \autoref{lem:no-inclusion}.\qed
\end{itemize}
\end{proof}

\paragraph{Proof of \autoref{lem:no-inclusion}.}

By contradiction assume there are $P, Q \in \mathcal{P}$ with $Q \subset P$. Let $P^*$ be such that $|P^*|$ is minimal among all possible choices of such a $P$.

For $Q \subset P^*$ define $b(Q) \in Q$ to be the maximal element w.r.t.~$<_Q$ such that $p <_{P^*} b(Q)$ for all $p \in (Q, b(Q))$, i.e., until element $b(Q)$ the order of $Q$ is identical to that of $P$. By \autoref{rem:order}, $b(Q)$ cannot be the last element of $Q$. So let $a(Q) \in Q$ be the successor of $b(Q)$ in $Q$. Note that this implies $a <_{P*} b$ by definition of $b(Q)$.
Among all paths $Q \subset P^*$, choose $Q^*$ such that $b^* := b(Q^*)$ is maximal w.r.t.~$<_{P^*}$. Let $a^* := a(Q^*)$. 

Let $R := Q^* \times_{b^*} P^*$. Note that $a^* \notin R$, as $a^* >_{Q^*} b^*$, and therefore $R \subset P^*$. We now claim that $<_R$ is identical to $<_{Q^*}$ on the $(Q^*, b^*)$-part of $R$.

\begin{claim}
For all $c, d \in R \cap (Q^*, b^*)$ with $c <_{Q^*} d$, we have $c <_R d$.
\end{claim}
\begin{proof}
If $c <_{Q^*} d$ but $d <_{R} c$, let $R' := R \times_d Q^*$. Note that $c \notin R'$ and by \autoref{ass:cross}, we have chosen $R$ such that $[R, d] \subset Q^*$. Thus $R' \subset Q^* \subset P^*$ which contradicts the choice of $P^*$.\qed
\end{proof}

By definition of $b(Q^*)$, the order $<_{Q^*}$ is identical to $<_{P^*}$ on $(Q^*, b^*)$ and thus $<_R$ is identical to $<_{P^*}$ on the $(Q^*, b^*)$-part of $R$. This implies that $a(R), b(R)$ cannot be both in the $(Q^*, b^*)$-part of $R$. Thus, $a(R) \in [b^*, P^*]$, which by $a(R) <_{P^*} b(R)$ implies that $b(R) \in (b^*, P^*)$. However this means $b(R) >_{P^*} b^*$ contradicting our choice of $Q^*$ maximizing $b^*$.\qed

\section{Proof of \autoref{lem:cut-over-time}}\label{sec:final-proof}
We will show that $C$ not only covers all paths in $\pathsT$ but even those paths that use waiting at intermediate elements, implying optimality of the constructed temporally repeated abstract flow for the relaxation of the problem. We are thus allowed to use the results from \autoref{sec:structure} in the proof, which is only sketched here (a complete proof can be found in the appendix).
\vspace{-0.5cm}
\subsubsection{\autoref{lem:cut-over-time}a.}
\emph{$C \cap P_\sigma \neq \emptyset$ for every $P_\sigma \in \pathsTw$.}

\begin{proof}
By contradiction assume there is a path that is not covered by $C$. Among all uncovered paths choose $P_\sigma \in \pathsTw$ such that $\sum_{e \in P} (\tau(e) + \tilde{y}(e))$ is minimal. We will show that there is an uncovered path $R$ whose length is strictly shorter, yielding a contradiction.

Let $\bar{e} \in P$ be maximal w.r.t.~$<_P$ among all element on $P$ with $\arr(P_\sigma, \bar{e}) \geq \alpha(\bar{e})$. Note that such an element exists because the first element $e_0$ of $P$ fulfills $\arr(P_\sigma, e_0) = \sigma(e_0) \geq 0 = \alpha(e_0)$. By construction, $P_\sigma$ arrives at $\bar{e}$ after the element has entered the cut. Note that, as $P_\sigma$ is not covered by the cut, the path must actually arrive at $\bar{e}$ after it has left the cut again, i.e., $\arr(P_\sigma, \bar{e}) \geq \alpha(\bar{e}) + \tilde{y}(\bar{e})$. Adding $\tau(\bar{e})$ to both sides of the inequality yields
\begin{equation}
\label{eqn:after_cut}
\sum_{e \in [P, \bar{e}]} (\sigma(e) + \tau(e)) ~\geq~ \alpha(\bar{e}) + \tilde{y}(\bar{e}) + \tau(\bar{e}). 
\end{equation}

We will now argue that $\bar{e}$ cannot be the final element of $P$. Assume by contradiction this was the case. Then let $L$ be a path with $\alpha(\bar{e}) = \sum_{e \in (L, \bar{e})} (\tau(e) + \tilde{y}(e))$. As $L \times_{\bar{e}} P \subseteq [L, \bar{e}]$, feasibility of $\tilde{y}$ implies $\alpha(\bar{e}) + \tau(\bar{e}) + \tilde{y}(\bar{e}) = \sum_{e \in [L, \bar{e}]} (\tau(e) + \tilde{y}(e)) \geq T$. Combining this with \eqref{eqn:after_cut} yields $\sum_{e \in P} (\sigma(e) + \tau(e)) \geq T$, a contradiction to $P_\sigma \in \pathsT$.

Thus, $\bar{e}$ is not the final element of $P$ and so we let $e'$ be the successor of $\bar{e}$ on $P$. 
Observe that by choice of $\bar{e}$ and definition of $\alpha$
\begin{equation*}
\arr(P_\sigma, e')
~<~ \alpha(e')
~\leq~ \sum_{e \in [P, \bar{e}]} (\tau(e) + \tilde{y}(e)).
\end{equation*}

Note that the left hand side of \eqref{eqn:after_cut} is at most $\arr(P_\sigma, e')$ and thus combining both inequalities yields
\begin{equation}
 \alpha(\bar{e}) < \sum_{e \in (P, \bar{e})} (\tau(e) + \tilde{y}(e)).
\end{equation}

Now let $Q \in \mathcal{P}$ be a path with $\sum_{e \in (Q, \bar{e})} (\tau(e) + \tilde{y}(e)) = \alpha(\bar{e})$. We consider the path $R := Q \times_{\bar{e}} P$. Observe that 
$$\sum_{e \in R} (\tau(e) + \tilde{y}(e)) \leq \alpha(\bar{e}) + \sum_{e \in (\bar{e}, P)} (\tau(e) + \tilde{y}(e)) < \sum_{e \in P} (\tau(e) + \tilde{y}(e)).$$

Let $s := \sum_{e \in [Q, \bar{e}]} \tilde{y}(e) + \sum_{e \in [Q, \bar{e}] \setminus R} \tau(e)$ and $\sigma' := (s, 0, \ldots, 0) \in \T^R$. We will show that the time expanded path $R_{\sigma'}$ is not covered by $C$, which contradicts the choice of $P$ as uncovered path minimizing the length w.r.t.~$\tau + \tilde{y}$.

Let $f \in R$. We show $(f, \arr(R, f)) \notin C$. Note that $\arr(R, f) = s + \sum_{e \in (R, f)} \tau(e)$, and further that by our results from \autoref{sec:structure} we can know that $R$ consists of two parts: The first part containing elements form $[Q, \bar{e}]$ in the same order as $<_Q$, the second part containing elements form $(\bar{e}, P)$, in the same order as $<_P$.
\begin{itemize}
\item If $f \in [Q, \bar{e}]$,
$$\arr(R, f) ~\geq~ \sum_{e \in (Q, f)} (\tau(e) + \tilde{y}(e)) + \tilde{y}(f) ~\geq~ \alpha(f) + \tilde{y}(f).$$
So $R_{\sigma'}$ reaches $f$ after it has left the cut in this case.
\item If $f \in R \setminus [Q, \bar{e}]$,
\begin{eqnarray*}
\arr(R, f) & = & \sum_{e \in [Q, \bar{e}]} (\tau(e) + \tilde{y}(e)) + \sum_{\mathrlap{e \in (\bar{e}, R) \cap (R, f)}} \tau(e) \ \leq \  \alpha(\bar{e}) + \tau(\bar{e}) + \tilde{y}(\bar{e}) + \sum_{\mathrlap{e \in (\bar{e}, P) \cap (P, f)}}\; \tau(e)\\
&\leq & \arr(P, \bar{e}) + \sum_{\mathrlap{e \in [\bar{e}, P] \cap (P, f)}} (\tau(e) + \sigma(e)) \ \leq \ \sum_{e \in (P, f)} (\tau(e) + \sigma(e)) + \sigma(f).
\end{eqnarray*}
The last term is strictly less than $\alpha(f)$ due to the choice of $\bar{e}$ and the fact that $f >_P \bar{e}$. So $R_{\sigma'}$ reaches $f$ before it enters the cut in this case.
\end{itemize}
This concludes the proof.\qed
\end{proof}

\section{Conclusion}
We presented abstract flows over time, an extension of flows over time that can be viewed as a first approach towards more general dynamic packing IPs. Our main result shows that the max flow/min cut result of Ford and Fulkerson still is valid in Hoffman's setting of abstract flows, emphasizing the robustness of the concept. At their heart, our proofs relied exclusively on the switching axiom for abstract networks, showing how this abstraction actually captures the essence of total dual integrality in network-based packing problems.

\bibliographystyle{splncs03}
\bibliography{abstract_flows_over_time}

\begin{thebibliography}{10}
\providecommand{\url}[1]{\texttt{#1}}
\providecommand{\urlprefix}{URL }

\bibitem{fleischer2003}
Fleischer, L., Skutella, M.: Minimum cost flows over time without intermediate
  storage. In: Proceedings of the fourteenth annual {ACM-SIAM} symposium on
  Discrete algorithms. pp. 66--75 (2003)

\bibitem{fleischer1998}
Fleischer, L., Tardos, E.: Efficient continuous-time dynamic network flow
  algorithms. Operations Research Letters  23(3-5),  71--80 (1998)

\bibitem{ford54}
Ford, L., Fulkerson, D.: Maximal flow through a network  (1954)

\bibitem{ford62}
Ford, L., Fulkerson, D.: Flows in networks. Princeton University Press (1962)

\bibitem{gaillard97}
Gaillard, A.: Switchdec polyhedra. Discrete applied mathematics  76(1),
  141--163 (1997)

\bibitem{gross2012}
Gro\ss, M., Skutella, M.: Maximum multicommodity flows over time without
  intermediate storage. In: Algorithms - {ESA} 2012. pp. 539--550 (2012)

\bibitem{hall2007}
Hall, A., Hippler, S., Skutella, M.: Multicommodity flows over time: Efficient
  algorithms and complexity. Theoretical Computer Science  379(3),  387--404
  (2007)

\bibitem{hoffman74}
Hoffman, A.: A generalization of max flow-min cut. Mathematical Programming
  6(1),  352--359 (1974)

\bibitem{hoffman1978}
Hoffman, A., Schwartz, D.: On lattice polyhedra. In: Proceedings 5th Hungarian
  Coll. on Combinatorics, North Holland. pp. 593--598 (1978)

\bibitem{hoppe2000}
Hoppe, B., Tardos, E.: The quickest transshipment problem. Mathematics of
  Operations Research  25(1),  36--62 (2000)

\bibitem{klinz2004}
Klinz, B., Woeginger, G.: Minimum-cost dynamic flows: The series-parallel case.
  Networks  43(3),  153--162 (2004)

\bibitem{martens08}
Martens, M., McCormick, S.T.: A polynomial algorithm for weighted abstract
  flow. Integer Programming and Combinatorial Optimization pp. 97--111 (2008)

\bibitem{mccormick96}
McCormick, S.T.: A polynomial algorithm for abstract maximum flow. In:
  Proceedings of the seventh annual ACM-SIAM symposium on Discrete algorithms.
  pp. 490--497. Society for Industrial and Applied Mathematics (1996)

\bibitem{schrijver1984}
Schrijver, A.: Total dual integrality from directed graphs, crossing families
  and sub-and supermodular functions. Progress in combinatorial optimization
  pp. 315--361 (1984)

\bibitem{skutella09}
Skutella, M.: An introduction to network flows over time. In: Cook, W.J.,
  Lov\'asz, L., Vygen, J. (eds.) Research Trends in Combinatorial Optimization,
  pp. 451--482. Springer (2009)

\end{thebibliography}

\end{document}